\documentclass[journal]{IEEEtran}

\usepackage{amsmath, amsfonts, amssymb}
\usepackage{amsthm}
\usepackage{bm}

\usepackage{algorithm}
\usepackage{algpseudocode}

\usepackage{graphicx}
\graphicspath{{figure/}}
\usepackage{xcolor}
\usepackage[caption=false, font=footnotesize]{subfig}
\usepackage{multirow}
\usepackage{booktabs}
\usepackage{capt-of}  

\usepackage{cite}
\usepackage[hidelinks]{hyperref}

\newtheorem{theorem}{Theorem}
\newtheorem{lemma}[theorem]{Lemma}
\newtheorem{proposition}[theorem]{Proposition}

\newtheorem{corollary}[theorem]{Corollary}
\newtheorem*{remark}{Remark}

\usepackage{xparse}  

\newif\ifshowtodo
\showtodotrue

\title{Switching-Reference Voltage Control for Distribution Systems with AI-Training Data Centers}

\author{Mingyuan~Yan,
        Trager~Joswig-Jones,
        Baosen~Zhang,
        Yize~Chen,
        and~Wenqi~Cui%
\thanks{M. Yan and W. Cui are with the Department of Electrical and Computer Engineering, New York University, New York, NY 10012, USA (e-mail: my3188@nyu.edu; wenqicui@nyu.edu).}%
\thanks{T. Joswig-Jones and B. Zhang are with the Department of Electrical and Computer Engineering, University of Washington, Seattle, WA 98195, USA (e-mail: joswitra@uw.edu; zhangbao@uw.edu).}%
\thanks{Y. Chen is with the Department of Electrical and Computer Engineering, University of Alberta, Edmonton, AB T6G 2R3, Canada (e-mail: yize.chen@ualberta.ca).}%
\thanks{This work was partially supported by the Henry Luce Foundation.}}

\begin{document}
\maketitle

\begin{abstract}

Large-scale AI training workloads in data centers exhibit rapid and periodic power swings that can induce voltage deviations in power distribution systems. Existing voltage controllers treat these swings as a generic disturbance, leading to high control effort but still large voltage violations. However, such emerging loads are not random: they alternate between two distinct operating phases. This paper exploits this structure with a decentralized switching-reference voltage control framework. By switching each voltage reference in step with the workload phases, the controller cancels the phase-induced voltage shift, holding the voltage within limits with low control effort. Because real-time communication between buses is not always available, the controller is designed to infer the reference from local voltage measurements. This paper further proves convergence under deadband and saturation. In case studies on real AI training power traces, the switching reference suppresses voltage violations, sometimes eliminating them entirely, while reducing the control effort by approximately an order of magnitude compared with conventional droop control. Further experiments confirm that it remains effective with multiple data centers and internal load smoothing.
\end{abstract}


\vspace{-0.5cm}
\section{Introduction}

\IEEEPARstart{T}{he} rapid growth of artificial intelligence (AI) has led to the deployment of increasingly large training workloads in modern data centers.
Empirical scaling laws show that model performance improves predictably with increasing model size, data, and computational budget~\cite{kaplan2020scaling, hoffmann2022training}.
In turn, electricity demand from AI training has also grown rapidly~\cite{iea_energy_and_ai}.
Consequently, the scale of modern data centers makes their power consumption an increasingly important factor in power system planning and operation.

\begin{figure}[!t]
    \centering
    \includegraphics[width=0.9\columnwidth]{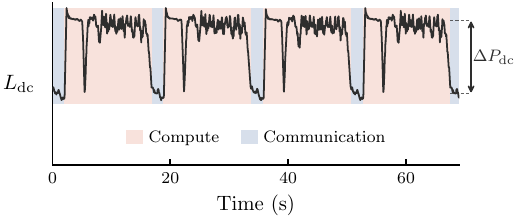}
    \vspace{-0.3cm}
    \caption{
    Data center power consumption $L_{\mathrm{dc}}$ from an at-scale training job on DGX-H100 racks~\cite{choukse2025power}.
    The trace alternates between high-power compute phases and low-power communication phases.
    These abrupt phase transitions are the load structure that the proposed switching-reference controller exploits.
    \vspace{-0.6cm}
    }
    \label{fig:train_power}
\end{figure}

One challenge introduced by data centers is the highly dynamic load behavior associated with large-scale AI training.
 Unlike conventional demand, which typically evolves gradually over time, AI training loads are synchronized and produce large active power swings within seconds in a data center's aggregate demand. These swings can reach up to hundreds of megawatts~\cite{grattafiori2024llama, chen2025electricity}. A representative power trace is shown in Fig.~\ref{fig:train_power}, drawn from an at-scale training job reported in a recent study~\cite{choukse2025power}. In a distribution system, these power swings can cause large voltage deviations.
Without effective regulation, voltages may exceed their allowable operating limits, challenging the operation of both data centers and distribution networks.

Voltage control methods have been widely used to reduce voltage violations in power distribution systems. Mechanical devices
such as tap-changing transformers can adjust voltage, but they are unsuitable for frequent adjustment due to equipment lifespan limitations~\cite{agalgaonkar2014distribution}.
Inverter-based resources such as photovoltaic and battery storage inverters can instead update reactive power on sub-second timescales, and have emerged as valuable resources for fast voltage support. Building on this capability, voltage droop control and its variants are proposed to regulate reactive power proportionally with voltage deviations~\cite{zhu2015fast, bernstein2018network, ieee1547, cui2025leveraging}, with established guarantees on stability and performance. Many other methods have been proposed to optimize or learn the control response, including model predictive control with explicit forecasts~\cite{valverde2013mpc}, reinforcement learning with offline-trained policies~\cite{zhang2021drlvoltvar}, and online feedback optimization driven by real-time measurements~\cite{mao2026feedback}.

Complementary to grid-side control, data center power smoothing methods have been explored using internal infrastructure such as workload-aware power capping~\cite{sakalkar2020data}, server-level power modulation~\cite{hou2020decentralized, chen2025voltage}, coordinated uninterruptible power supply (UPS) and cooling control~\cite{xie2025enhancing}, and accelerator utilization control~\cite{liang2026gpu}.  However, existing approaches typically either overlook the underlying power dynamics of AI training workloads, or treat these power swings as unstructured disturbances. As a result, they may require unnecessarily high control effort while providing limited improvement in worst-case voltage deviations.

The ineffectiveness of existing control laws is caused by the switching between distinct phases in the power consumption of AI training loads. 
As shown in Fig.~\ref{fig:train_power}, AI training workloads exhibit a two-phase pattern, alternating between a high-power compute phase and a low-power communication phase, with an abrupt transition between them. 
Since the most severe voltage violations arise during these abrupt transitions,
existing regulation methods designed for smoothly varying loads
may not respond sufficiently fast to mitigate the worst voltage violations.
In addition,
responding to large power swings incurs high control effort.
In this work,
we aim to answer the following question:

\emph{How should we design a voltage control law to mitigate voltage deviations induced by AI training workloads,
while avoiding excessive control effort?}


The proposed solution is to exploit the structured switching dynamics of AI training loads in the design of a decentralized switching-reference voltage controller for distribution systems.
The key idea is to proactively switch the control references to prepare for possible abrupt load increases or decreases, rather than relying solely on reflexive actions after voltage deviations have already occurred.
More concretely, rather than trying to drive the voltage to the nominal operating point, we dynamically choose more suitable voltage references.
To this end,
we first characterize the two-phase structure of large-scale AI training workloads,
and then design a switching-reference droop controller
that aligns the voltage reference with the current phase of the data center load.
On this basis, we establish conditions for voltage convergence, and show how to choose the reference levels to cancel the phase-induced disturbance and place the bias to minimize the worst-case voltage deviation.
Further, we develop a rule for switching voltage references based solely on local voltage measurements,
enabling simple local implementation
while significantly reducing control effort.
Unlike centralized or coordinated schemes, the controller requires no communication between buses.
The analysis further extends to a practical implementation with deadband and action saturation.
Case studies with real data show that
the proposed approach substantially reduces voltage violations and control effort
compared with conventional fixed-reference droop control.
It applies to multiple data centers, generalizes across diverse AI training power traces, and is compatible with internal data center load smoothing.

\section{Network Model and Problem Formulation}

\subsection{Notation}

Throughout the paper, $|\cdot|$ denotes elementwise absolute value; vector inequalities, $\max$, $\min$, $\limsup$, $\liminf$, and elementwise matrix bounds are interpreted elementwise. Matrix orderings $\succ,\prec$ for symmetric matrices follow the positive-definite convention. Vectors are denoted in lower-case bold and matrices are denoted in upper-case bold, and scalars are unbolded unless otherwise specified. 
Subscript $i$ indicates variables for bus $i$.
For scalars $a_{1,t}, \cdots, a_{n,t}$, each associated with one of the $n$ buses at time $t$, $\bm{a}_t:=(a_{1,t},\cdots, a_{n,t})^\top$ stacks them into a column vector. For a matrix $\mathbf M$, $M_{ij}$ denotes its $(i,j)$ entry and $\bm m_j$ its $j$-th column.

\subsection{Network Model and Data Center Load Characteristics}
\label{sec:dc_network}

We consider a radial distribution system with $n$ buses and a feeder. Let $\bm v_t, \bm p_t, \bm q_t \in \mathbb{R}^n$ denote the vectors of bus voltage magnitudes, net active power injections, and net reactive power injections at time $t$ (positive for generation, negative for load). Following common practice, voltages are normalized so that the nominal value at all buses is $1$ per unit (p.u.).
We adopt the LinDistFlow model in~\cite{baran1989network, zhu2015fast} to characterize voltages affected by active and reactive power across the distribution system
\vspace{-0.1cm}
\begin{equation}
\bm v_t = \mathbf R \bm p_t + \mathbf X \bm q_t + \mathbf{1},
\label{eq:lin_voltage}
\end{equation}
 where $\mathbf{1}$ is the all-ones vector and $\mathbf R, \mathbf X \in \mathbb{R}^{n\times n}$ are symmetric positive definite matrices with nonnegative elements; $\mathbf R$ is built from the line resistances and $\mathbf X$ from the line reactances of the network. 

Suppose an AI training data center is connected at bus $j$, with its net power injection at time $t$ denoted by $p_{j,t}$.
Such workloads proceed as a sequence of training iterations,
each consisting of two phases as illustrated in Fig.~\ref{fig:train_power}:
a \emph{high-power compute phase}, during which accelerators (e.g., GPUs, TPUs, and Trainium) process a mini-batch of training data through forward and backward propagation,
and a \emph{low-power communication phase}, during which accelerators reduce gradients and broadcast parameters across the cluster, and occasionally write model checkpoints to storage.
Unlike conventional slowly varying demand, the aggregate data center load follows this two-phase pattern rather than averaging out, because all accelerators in a job execute the same iteration synchronously~\cite{choukse2025power}. The accelerators dominate a training server's power and drive these fast two-phase swings~\cite{patel2024characterizing}, while cooling and other loads vary slowly and enter our model additively. Accordingly, we model the net injection at this bus as
\vspace{-0.1cm}
\begin{equation}
p_{j,t} = \bar p_j(\phi_t) + w_t ,
\end{equation}
where $\phi_t\in\{0,1\}$ indicates the phase, $0$ in communication and $1$ in compute; $\bar p_j(\phi_t)$ is the phase-mean injection, the average injection in phase $\phi_t$, combining the slowly varying contribution of other loads with the phase-mean data center load; and $w_t$ is the intra-phase fluctuation. The data center draws more power in compute, so $\bar p_j(1) < \bar p_j(0) < 0$. During a phase transition, $p_{j,t}$ swings from one level to the other, an abrupt change of magnitude $\Delta P_{\mathrm{dc}} := |\bar p_j(1)-\bar p_j(0)|$. These large and rapid swings cause voltage deviations across the network, which can exceed the admissible $\pm5\%$ band around nominal~\cite{national1996american}. Throughout this paper, we assume that active power variations at buses other than the data center are small relative to those at the data center. Accordingly, $p_{i,t}$ is treated as constant for all $i \neq j$.

To reduce these voltage violations, we can regulate voltage through the reactive power injections $\bm q_t$, supplied by inverter-based resources deployed at multiple buses across the network~\cite{ieee1547}. These resources, which may include battery energy storage systems and photovoltaic inverters, are not restricted to the buses connected with data centers. At the same time, the controller should keep the control effort low by avoiding large or frequent control actions. Beyond reducing voltage violations with low control effort, we seek a simple and practical controller with the following properties:
\begin{itemize}
\item \emph{Decentralized:} each bus sets its reactive power injection from only its local voltage $v_{i,t}$, since real-time communication between buses is often not feasible.
\item \emph{Structural:} it relies only on the two-phase structure of the load, namely the swing between two distinct levels. It should not depend on the exact switching times, transition magnitude $\Delta P_{\mathrm{dc}}$, or noise $w_t$, and does not rely on predicting them.
\item \emph{Adaptive:} the same controller applies to any load with this structure, without retuning.
\end{itemize}

\subsection{Voltage Control and Challenges}
\label{sec:challenges}
A standard decentralized controller is the Volt--VAR droop control and its variants~\cite{ieee1547, bernstein2018network, cui2022decentralized}, where reactive power injection is adjusted linearly with the deviation of local voltage $v_{i,t}$ from its nominal value:
\vspace{-0.1cm}
\begin{equation}
\bm{q}_{t+1} = \bm{q}_t - \mathbf{K} \bigl(\bm{v}_t - \bm{1}\bigr),
\label{eq:droop_fixed}
\end{equation}
where $\mathbf{K}=\mathrm{diag}(k_1,\ldots,k_n)\succ 0$ and each $k_i$ is the droop coefficient at bus $i$.

This droop control is \emph{reflexive}: it acts only after a voltage deviation has appeared, not before the load change that causes it. The phase transitions of AI training workloads are abrupt, so the voltage deviation reaches its largest value before the controller can respond. Since the phase transitions are frequent, these large voltage deviations persist.

By contrast, the proposed controller is \emph{proactive}, as it exploits the two-phase structure of AI training workloads to switch the voltage reference between the two phases. The controller relies only on real-time local voltage measurements, not on predicted load or voltage.
As we will show, this approach maintains voltages within the desired range with smaller and less frequent control actions.

\section{Switching-Reference Control}
\label{sec:grid-side}

\subsection{Switching-Reference Control Law}

\begin{figure}[!t]
    \centering
    \includegraphics[width=0.85\columnwidth]{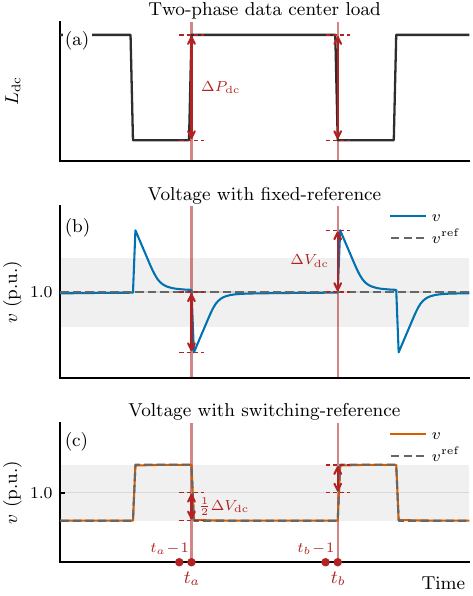}
    \vspace{-0.2cm}
    \caption{
    Voltage response at the data center bus to phase transitions in load. The load step $\Delta P_{\mathrm{dc}}$ produces a voltage step $\Delta V_{\mathrm{dc}}=R_{jj}\Delta P_{\mathrm{dc}}$, with $R_{jj}$ the resistance from the substation to the data center bus. 
    (a) Data center load alternating between the two phases.
    (b) Fixed-reference droop ($v^{\mathrm{ref}}\equiv 1$): each transition induces a voltage shift of $\approx \Delta V_{\mathrm{dc}}$.
    (c) Switching-reference droop: the worst post-transition voltage deviation reduces to $\approx\tfrac12 \Delta V_{\mathrm{dc}}$.
    Aligning the voltage reference with the load phase halves the worst-case voltage deviation relative to the fixed reference.
    \vspace{-0.5cm}
    }
    \label{fig:power_voltage_causality}
\end{figure}

The proposed controller takes the form
\vspace{-0.1cm}
\begin{equation}
\bm{q}_{t+1} = \bm{q}_t - \mathbf{K} \bigl(\bm{v}_t - \bm{v}_t^{\mathrm{ref}}\bigr),
\label{eq:droop}
\end{equation}
which differs from~\eqref{eq:droop_fixed} in that the feedback reference $\bm{v}_t^{\mathrm{ref}}$ is time-varying.
To see why this reference is designed to be time-varying, we substitute the controller~\eqref{eq:droop} into the LinDistFlow model~\eqref{eq:lin_voltage} and equivalently represent the transition of the voltage deviation as follows
\[
\bm v_{t+1}-\bm v_{t+1}^{\mathrm{ref}}
=
(\mathbf I-\mathbf X\mathbf K)\bigl(\bm v_t-\bm v_t^{\mathrm{ref}}\bigr)
+\bigl(\mathbf R\,\Delta\bm p_t-\Delta\bm v_t^{\mathrm{ref}}\bigr),
\]
where $\Delta\bm p_t:=\bm p_{t+1}-\bm p_t$ and $\Delta\bm v_t^{\mathrm{ref}}:=\bm v_{t+1}^{\mathrm{ref}}-\bm v_t^{\mathrm{ref}}$. The detailed derivation is given
in Appendix~\ref{app:error_dynamics}. A load change $\Delta\bm p_t$ induces voltage deviations of the form $\mathbf R\Delta\bm p_t$. Our key design is to switch the voltage reference so that $\Delta\bm v_t^{\mathrm{ref}}$ approximately cancels the load-induced deviation $\mathbf R\Delta\bm p_t$. In the following, we first illustrate the intuition behind this cancellation using a simplified two-phase data center load model. We then present the full formulation for constructing the time-varying voltage reference.


Consider the idealized data center load in Fig.~\ref{fig:power_voltage_causality}(a),
where intra-phase fluctuations have been smoothed out and the data center switches between the two phases with transition magnitude
$\Delta P_{\mathrm{dc}}$.
This load change $\Delta P_{\mathrm{dc}}$ causes a voltage step change $\Delta V_{\mathrm{dc}}:=R_{jj}\Delta P_{\mathrm{dc}}$, where $R_{jj}$ is the resistance of the path from the substation to the data center bus $j$.
Conventional fixed-reference droop ($\bm v_t^{\mathrm{ref}}\equiv\bm 1$) is reflexive:
it can only act after a voltage deviation has appeared.
The full voltage step change $\Delta V_{\mathrm{dc}}$ then appears as the largest voltage deviation, as in Fig.~\ref{fig:power_voltage_causality}(b).
This spike of voltage deviations is independent of the droop gain $\mathbf{K}$;
increasing $\mathbf{K}$ accelerates its decay but does not reduce its peak.

The proposed controller is proactive: it switches the reference between two levels to keep the voltage deviation small, instead of letting a large deviation appear and then correcting it. To convey the intuition behind the proposed design, let $t_a$ and $t_b$ be the two transition instants in a cycle, the communication-to-compute and compute-to-communication transitions. At $t_a$, the load at bus $j$ increases by $\Delta P_{\mathrm{dc}}$. At $t_a-1$, still in the communication phase, the reference is at the communication-phase level $v_{j,t_a-1}^{\mathrm{ref}}=1+\tfrac12 \Delta V_{\mathrm{dc}}$, and the voltage matches this reference. The transition then causes the full voltage step decrease $\Delta V_{\mathrm{dc}}$, so the bus voltage at $t_a$ only reaches
\vspace{-0.2cm}
\begin{equation*}
v_{j,t_a}\approx\underbrace{\Bigl(1+\tfrac12 \Delta V_{\mathrm{dc}}\Bigr)}_{v_{j,t_a-1}^{\mathrm{ref}}}-\underbrace{\Delta V_{\mathrm{dc}}\vphantom{\Bigl(\Bigr)}}_{\text{step change}}=1-\tfrac12 \Delta V_{\mathrm{dc}}.
\end{equation*}
After detecting such a step voltage change, the voltage reference at time step $t_a$ is switched to $v_{j,t_a}^{\mathrm{ref}}=1-\tfrac12 \Delta V_{\mathrm{dc}}$, and this reference is maintained until the next step voltage change is detected at time step $t_b$. The subsequent compute-to-communication transition at $t_b$ is symmetric. At $t_b-1$, the reference is at the compute-phase level $v_{j,t_b-1}^{\mathrm{ref}}=1-\tfrac12 \Delta V_{\mathrm{dc}}$. The load then decreases by $\Delta P_{\mathrm{dc}}$, so the voltage steps up by $\Delta V_{\mathrm{dc}}$ to $1+\tfrac12 \Delta V_{\mathrm{dc}}$, which equals the communication-phase voltage reference. This voltage step change is thus split evenly around 1 p.u., as illustrated by Fig.~\ref{fig:power_voltage_causality}(c). The worst voltage deviation is therefore half that of the fixed reference, within $\pm\tfrac12 \Delta V_{\mathrm{dc}}$. Because the voltage follows the switching reference, the droop needs smaller control actions, rather than the high control effort resulting from a large voltage deviation.

Motivated by the above intuition, we therefore implement the time-varying reference as a switching reference that alternates between the compute and communication phases.
More formally, we parameterize the reference as
\vspace{-0.2cm}
\begin{equation}
v_{i,t}^{\mathrm{ref}} = 1 + s_{i,t}  v_i^{\mathrm{amp}},
\label{eq:tv_ref_nobias}
\end{equation}
where $s_{i, t}\in\{-1,+1\}$ selects one of the two reference levels,
and $v_i^{\mathrm{amp}}\in\mathbb{R}_{\ge 0}$ determines their magnitude.
The sign $s_{i,t}=+1$ corresponds to the communication phase ($\phi_t=0$) and gives the higher reference $1+v_i^{\mathrm{amp}}$, while $s_{i,t}=-1$ corresponds to the compute phase ($\phi_t=1$) and gives the lower reference $1-v_i^{\mathrm{amp}}$.
We will later show how to set $s_{i,t}$ from local voltage measurements to track the phase $\phi_t$.

In practice,
voltage fluctuations exist within each phase
and may differ in magnitude between the two phases, for instance larger in the compute phase than in the communication phase.
To accommodate this asymmetry,
we introduce a bias term so the two reference levels need not be symmetric about $1$~p.u.
Accordingly, we generalize~\eqref{eq:tv_ref_nobias} to
\begin{equation}
v_{i,t}^{\mathrm{ref}}
=
v_{i}^{\mathrm{bias}}
+
s_{i,t}  v_i^{\mathrm{amp}},
\label{eq:tv_ref}
\end{equation}
where
$v_i^{\mathrm{bias}}\in\mathbb{R}$ is the bias of voltage reference at bus $i$.

Under this parameterization,
the two reference levels for each bus $i$ are centered at $v_{i}^{\mathrm{bias}}$ and separated by $2v_i^{\mathrm{amp}}$.
The remaining question is how to set these reference levels at the different buses to match the phase-induced voltage shift.

\subsection{Convergence Analysis and Reference Design}
\label{sec:deviation_analysis}

We now analyze the closed-loop deviation from the reference and design the references to reduce the worst-case voltage deviation from 1 p.u.

Writing the load and reference changes as the disturbance $\bm d_t := \mathbf R\,\Delta\bm p_t-\Delta\bm v_t^{\mathrm{ref}}$,
the recursion of voltage 
deviation is
\begin{equation}
\bm{v}_{t+1}-\bm{v}_{t+1}^{\mathrm{ref}}
=
(\mathbf I - \mathbf X \mathbf K)\bigl(\bm{v}_t-\bm{v}_t^{\mathrm{ref}}\bigr)
+
\bm d_t.
\label{eq:linear_error_system}
\end{equation}
Although $\Delta\bm v_t^{\mathrm{ref}}$ involves $\bm v_{t+1}^{\mathrm{ref}}$, it appears only in the analysis: the implemented controller~\eqref{eq:droop} is causal and does not use future values.

The deviation from the reference $\bm v_t-\bm v_t^{\mathrm{ref}}$ is governed by the closed-loop matrix $\mathbf I-\mathbf X\mathbf K$ and the disturbance $\bm d_t$. We design $\mathbf K$ to make $\mathbf I-\mathbf X\mathbf K$ contractive, so this deviation decays when no persistent disturbance acts. 
The conditions are given in the following theorem:

\begin{theorem}[Convergence of voltage deviation]
\label{thm:deviation_bound}

If the diagonal gain matrix $\mathbf K$ satisfies
\(
0\prec \mathbf K \prec 2\mathbf X^{-1},
\)
then there exist a nonnegative matrix $\mathbf C$
and a constant $0<\epsilon<1$ such that
for any initial time $t_0$ and all $t\ge t_0$,
\begin{equation}
|\bm v_t-\bm v_t^{\mathrm{ref}}|
\le
\mathbf C\,\epsilon^{t-t_0}
|\bm v_{t_0}-\bm v_{t_0}^{\mathrm{ref}}|
+
\sum_{\tau=t_0}^{t-1}
\mathbf C\,\epsilon^{t-1-\tau}|\bm d_\tau|.
\end{equation}
In particular, for any  $\bar{\bm d}$ where
\(
|\bm d_t|\le \bar{\bm d}
\)
for all $t$, 
\begin{equation}
\limsup_{t\to\infty}\;
|\bm v_t-\bm v_t^{\mathrm{ref}}|
\le
\frac{1}{1-\epsilon}\mathbf C\,\bar{\bm d}.
\label{eq:deviation_bound_general}
\end{equation}
\end{theorem}

The proof is given in Appendix~\ref{app:deviation_bound}.
Theorem~\ref{thm:deviation_bound} shows that the voltage deviation $|\bm v_t-\bm v_t^{\mathrm{ref}}|$
is driven by the disturbance 
\(
\bm d_t
=
\mathbf R\Delta\bm p_t-\Delta\bm v_t^{\mathrm{ref}}.
\) For conventional fixed-reference droop control, where $\Delta\bm v_t^{\mathrm{ref}}=\bm 0$, $\bm d_t=\mathbf R\Delta\bm p_t$, so a large $\Delta\bm p_t$ results directly in a large voltage deviation.

Next, we show that a switching reference reduces $|\bm d_t|$ and hence the steady-state voltage deviation.
The active power increment at the data center can be decomposed as
\(
\Delta p_{j,t}
=
(\bar p_j(\phi_{t+1})-\bar p_j(\phi_t))
+
(w_{t+1}-w_t),
\)
where the first term represents the structured shift of power caused by phase transitions of data center loads,
and the second term captures the intra-phase fluctuations.
To cancel the voltage deviation induced by phase transitions,
we design the reference to depend only on the data center phase \(\phi_t\) rather than directly on time \(t\),
\vspace{-0.1cm}
\begin{equation}
\bm v_t^{\mathrm{ref}}
=
\bar{\bm v}^{\mathrm{ref}}(\phi_t),
\label{eq:ideal_tracking}
\end{equation}
where $\bar{\bm v}^{\mathrm{ref}}(0)$ and
$\bar{\bm v}^{\mathrm{ref}}(1)$ denote the reference voltages
associated with the two data center phases.
From~\eqref{eq:lin_voltage}, the voltage change induced by a data center
phase transition is
\(
\Delta\bm v_{\mathrm{phase}}
:=
\mathbf R\bigl(\bm{\bar p}(0)-\bm{\bar p}(1)\bigr),
\)
where $\bm{\bar p}(0), \bm{\bar p}(1)$ are the phase-mean injections in the two phases. Since only the data center load changes between phases, $\bm{\bar p}(0)-\bm{\bar p}(1) = (\bar p_j(0)-\bar p_j(1))\bm{e}_j$, so $\Delta\bm v_{\mathrm{phase}}$ acts only through $\bm r_j$, the column of $\mathbf R$ for bus $j$, which maps the data center's power change to the voltage shift at every bus.
The next proposition shows that a suitable choice of the two phase-dependent reference levels cancels the phase-switching disturbance, reducing $\bm d_t$ to the residual fluctuation $\bm r_j(w_{t+1}-w_t)$.
In the parametrization \eqref{eq:tv_ref}, this corresponds to
$\bm v^{\mathrm{bias}} = \bm b$ and
$\bm v^{\mathrm{amp}} = \tfrac{1}{2}\Delta \bm v_{\mathrm{phase}}$.

\begin{proposition}[Voltage reference switching]
\label{prop:mode_match_cancel}
Let the reference design follow \eqref{eq:ideal_tracking}
and $\Delta\bm v_{\mathrm{phase}} := \mathbf R\bigl(\bm{\bar p}(0)-\bm{\bar p}(1)\bigr)$.
If, for some bias vector $\bm b \in \mathbb{R}^{n}$, the two phase-dependent reference levels are chosen as
\vspace{-0.1cm}
\begin{equation}\label{eq:mode_match}
  \bm{\bar v}^{\mathrm{ref}}(0)
=
\bm b+\frac12\Delta\bm v_{\mathrm{phase}},
\quad \bm{\bar v}^{\mathrm{ref}}(1)
=
\bm b-\frac12\Delta\bm v_{\mathrm{phase}}, 
\end{equation}
then the disturbance $\bm d_t$ in~\eqref{eq:linear_error_system} reduces to
$
\bm d_t=\bm r_j(w_{t+1}-w_t).
$

\end{proposition}

\begin{proof}
Since $\Delta\bm p_t=\bm e_j\Delta p_{j,t}$, we have $\mathbf R\Delta\bm p_t=\bm r_j\bigl(\bar p_j(\phi_{t+1})-\bar p_j(\phi_t)+w_{t+1}-w_t\bigr)$. Under~\eqref{eq:ideal_tracking}, $\Delta\bm v_t^{\mathrm{ref}}=\bar{\bm v}^{\mathrm{ref}}(\phi_{t+1})-\bar{\bm v}^{\mathrm{ref}}(\phi_t)$, which by~\eqref{eq:mode_match} equals $\bm r_j\bigl(\bar p_j(\phi_{t+1})-\bar p_j(\phi_t)\bigr)$. Substituting both into $\bm d_t=\mathbf R\Delta\bm p_t-\Delta\bm v_t^{\mathrm{ref}}$, the phase-mean terms cancel and yield $\bm d_t=\bm r_j(w_{t+1}-w_t)$.
\end{proof}
Proposition~\ref{prop:mode_match_cancel} shows that
a proper choice of the references
can reduce the disturbance magnitude $|\bm d_t|$,
and consequently reduce the voltage deviation $|\bm v_t-\bm v_t^{\mathrm{ref}}|$.
Our ultimate objective, however, is to reduce the worst-case steady-state voltage deviation from 1 p.u.

In particular, under the reference design \eqref{eq:ideal_tracking}--\eqref{eq:mode_match},
the steady-state extrema $\bm v^+(\bm b)$ and $\bm v^-(\bm b)$
induced by the current bias $\bm b$ determine how to correct the bias to minimize the worst-case steady-state voltage deviation.

\begin{proposition}[Bias shifting]
\label{prop:minimax_deviation}

Under the reference design
\eqref{eq:ideal_tracking}--\eqref{eq:mode_match},
define for each phase $\phi\in\{0,1\}$ the steady-state voltage extrema
$\bm v^+(\phi,\bm b):=\limsup_{t\to\infty}\bm v_t$
and
$\bm v^-(\phi,\bm b):=\liminf_{t\to\infty}\bm v_t$ under fixed phase $\phi$.
The  extrema across two phases are given by
$\bm v^+(\bm b):=\max\{\bm v^+(0,\bm b),\,\bm v^+(1,\bm b)\}$
and
$\bm v^-(\bm b):=\min\{\bm v^-(0,\bm b),\,\bm v^-(1,\bm b)\}$.
Then the bias that minimizes the worst-case steady-state voltage deviation from  1 p.u. is
\vspace{-0.1cm}
\begin{equation}
\bm b^\star=\bm b+\Delta\bm b,\qquad \Delta\bm b=\bm 1-\frac{\bm v^+(\bm b)+\bm v^-(\bm b)}{2}.
\label{eq:b_star_from_extrema}
\end{equation}

\end{proposition}

\begin{proof}
At each bus, $v_i$ lies in $[v_i^-(\bm b),\,v_i^+(\bm b)]$, and the maximum distance from $1$ is attained at an endpoint. Thus, for each $i$, the worst-case steady-state voltage deviation is $\limsup_{t\to\infty}|v_{i,t}-1|=\max\{|v_i^+(\bm b)-1|,\;|v_i^-(\bm b)-1|\}$.

Now shift the bias from $\bm b$ to $\bm b+\Delta\bm b$. Since $\bm v^{\mathrm{ref}}(0)=\bm b+\tfrac12\Delta\bm v_{\mathrm{phase}}$ and $\bm v^{\mathrm{ref}}(1)=\bm b-\tfrac12\Delta\bm v_{\mathrm{phase}}$, both reference levels are translated by $\Delta\bm b$ while their difference remains unchanged. Hence $\Delta\bm v_t^{\mathrm{ref}}$ and the error dynamics in~\eqref{eq:linear_error_system} are unchanged, so the steady-state extrema shift by the same amount, $\bm v^+(\bm b+\Delta\bm b)=\bm v^+(\bm b)+\Delta\bm b$ and $\bm v^-(\bm b+\Delta\bm b)=\bm v^-(\bm b)+\Delta\bm b$. The worst-case steady-state voltage deviation from 1 p.u. then becomes $\max\{|\bm v^+(\bm b)+\Delta\bm b-\bm 1|,\;|\bm v^-(\bm b)+\Delta\bm b-\bm 1|\}$.
For each bus $i$, this is the maximum distance from $1$ to the interval $[v_i^-(\bm b)+\Delta b_i,\;v_i^+(\bm b)+\Delta b_i]$. Such distance  is minimized when $1$ is its midpoint, namely, $\tfrac{1}{2}(v_i^-(\bm b)+\Delta b_i+v_i^+(\bm b)+\Delta b_i)=1$, which yields~\eqref{eq:b_star_from_extrema}.
\end{proof}

Equation~\eqref{eq:b_star_from_extrema} gives the same $\bm b^\star$ for any $\bm b$, provided each phase lasts long enough for the voltage to settle to its steady-state extrema $\bm v^+$ and $\bm v^-$ (the $t\to\infty$ limits above). The correction $\Delta\bm b$ depends on the chosen $\bm b$, but $\bm b^\star$ does not.

\vspace{-0.2cm}
\subsection{Extension to Deadband and Saturation}

The analysis above works for the linear controller of~\eqref{eq:droop}. In practice, the controller additionally incorporates a per-bus voltage deadband and an action saturation, consistent with IEEE~1547~\cite{ieee1547}. We denote these operators by
\vspace{-0.2cm}
\begin{equation}
\begin{aligned}
\mathrm{dz}_{\bm\delta}(\bm x)_i &:= \mathrm{sign}(x_i)\max\{|x_i|-\delta_i,\,0\}, \\
\mathrm{sat}_{\bm\Delta}(\bm x)_i &:= \mathrm{sign}(x_i)\min\{|x_i|,\,\Delta_i\},
\end{aligned}
\label{eq:dz_sat_def}
\end{equation}
where $\bm\delta\in\mathbb R^n_{\ge 0}$ is the per-bus deadband and $\bm\Delta\in\mathbb R^n_{>0}$ is the action saturation. The implemented controller applies these operators to the linear update~\eqref{eq:droop}:
\vspace{-0.1cm}
\begin{equation}
\bm q_{t+1} = \bm q_t - \mathrm{sat}_{\bm\Delta}\!\bigl(\mathbf K\,\mathrm{dz}_{\bm\delta}(\bm v_t-\bm v_t^{\mathrm{ref}})\bigr).
\label{eq:db_sat_controller}
\end{equation}

These operators are nonsmooth and fall outside the linear analysis of Theorem~\ref{thm:deviation_bound}. The next result shows that the deviation bound remains finite and monotone in $\bar{\bm d}$ under both operators.

\begin{theorem}[Controller with deadband and saturation]
\label{thm:implemented_bound}
Let $(\mathbf C,\epsilon)$ be defined as in Theorem~\ref{thm:deviation_bound}. Suppose $0\prec\mathbf K\prec 2\mathbf X^{-1}$ and $|\bm d_t|\le\bar{\bm d}$ for all $t$, and assume the small-gain condition
\vspace{-0.2cm}
\begin{equation}
\rho(\mathbf C\mathbf X\mathbf K)<1-\epsilon.
\label{eq:small_gain}
\end{equation}
Then the asymptotic voltage deviation $\bar{\bm e}:=\limsup_{t\to\infty}|\bm v_t-\bm v_t^{\mathrm{ref}}|$ of the implemented controller satisfies
\begin{equation}
\begin{aligned}
\bar{\bm e}\;\le\;&\biggl(\mathbf I-\frac{\mathbf C\mathbf X\bm\Pi_S\mathbf K}{1-\epsilon}\biggr)^{\!-1}\\
&\cdot\biggl[\underbrace{\frac{\mathbf C\bar{\bm d}}{1-\epsilon}}_{\text{Theorem~\ref{thm:deviation_bound}}}+\underbrace{\frac{\mathbf C\mathbf X\mathbf K\bm\delta}{1-\epsilon}}_{\text{deadband}}-\underbrace{\frac{\mathbf C\mathbf X\bm\Pi_S(\mathbf K\bm\delta+\bm\Delta)}{1-\epsilon}}_{\text{saturation}}\biggr],
\end{aligned}
\label{eq:implemented_bound}
\end{equation}
where $S:=\{\,i:(\mathbf K(\bar{\bm e}-\bm\delta)-\bm\Delta)_i>0\,\}$ is the set of buses at which the action saturation remains asymptotically active and $\bm\Pi_S$ is the diagonal $0$/$1$ matrix with $(\bm\Pi_S)_{ii}=1$ iff $i\in S$.
\end{theorem}

The proof is given in Appendix~\ref{app:convergence}.
The bound is nonnegative and nondecreasing in $\bar{\bm d}$. By this monotonicity, the reference design of Section~\ref{sec:grid-side} continues to lower the voltage deviation under the implemented controller. The bound reduces to Theorem~\ref{thm:deviation_bound} when both operators are inactive and to Corollaries~\ref{cor:deadband} and~\ref{cor:saturation} in the deadband-only and saturation-only limits.

\section{Decentralized Implementation}
\label{sec:implementation}

\subsection{Measurement-Based Reference Adaptation}

The reference design of Section~\ref{sec:grid-side} requires the data center phase $\phi_t$, the phase-induced voltage shift $\Delta\bm v_{\mathrm{phase}}$, and the optimal bias $\bm b^\star$, none of which are directly available to the controller. Fortunately, both the phase-induced voltage shift $\Delta\bm v_{\mathrm{phase}}$ and the voltage extrema needed for bias correction can be approximated from recent voltage measurements. In particular, only local voltage observations are required for updating $v_i^{\mathrm{bias}}$ and $v_i^{\mathrm{amp}}$ for each bus $i$.

\begin{algorithm}[!t]
\caption{Decentralized switching-reference controller.}
\label{alg:controller}

\noindent\textbf{Parameters.}
\emph{Per bus:} droop gain $k_i > 0$, deadband $\delta_i \ge 0$, and action saturation $\Delta_i > 0$.
\emph{Shared:} amplitude and bias window lengths $N_a, N_b$, bias update period $D_b$, and smoothing gain $\eta_b \in (0, 1]$.

\noindent\textbf{Initialize.} For every bus $i$, set $v_i^{\mathrm{bias}} \gets 1$~p.u.; let $\mathcal{W}_i$ hold the most recent $N = \max(N_a, N_b)$ voltage samples, initially~empty.

\noindent\textbf{Procedure.} Each bus $i$ locally executes the following at every control instant $t = 1, 2, \ldots$:

\begin{algorithmic}[1]
\State Measure local voltage $v_{i,t}$
\State Push $v_{i,t}$ into $\mathcal{W}_i$; drop the oldest if $|\mathcal{W}_i| > N$
\State $v_i^{\mathrm{amp}} \gets \tfrac{1}{2}\max_{0\le\ell<N_a-1}|v_{i,t-\ell}-v_{i,t-\ell-1}|$
\State $s_i \gets \mathrm{sign}(v_{i,t}-v_i^{\mathrm{bias}})$
\State $v_i^{\mathrm{ref}} \gets v_i^{\mathrm{bias}} + s_i\,v_i^{\mathrm{amp}}$
\State $q_i \gets q_i - \mathrm{sat}_{\Delta_i}\!\bigl(k_i\,\mathrm{dz}_{\delta_i}(v_{i,t}-v_i^{\mathrm{ref}})\bigr)$
\If{$t \bmod D_b = 0$}
  \State $\Delta b \gets \tfrac{1}{2}\bigl(\max_{0\le\ell<N_b} v_{i,t-\ell} + \min_{0\le\ell<N_b} v_{i,t-\ell}\bigr) - 1$
  \State $v_i^{\mathrm{bias}} \gets v_i^{\mathrm{bias}} - \eta_b\,\Delta b$
\EndIf
\end{algorithmic}
\end{algorithm}

We therefore maintain a window $\mathcal{W}_i$ of the most recent $N = \max(N_a, N_b)$ voltage samples; the amplitude and bias updates use the most recent $N_a$ and $N_b$ of them, respectively. Specifically, for each bus $i$, we use the reference structure characterized in Propositions~\ref{prop:mode_match_cancel} and~\ref{prop:minimax_deviation}
to construct the time-varying voltage reference
\[
v_{i,t}^{\mathrm{ref}}
=
v_{i,t}^{\mathrm{bias}}
+
s_{i,t} v_{i,t}^{\mathrm{amp}},
\]
where $v_{i,t}^{\mathrm{bias}}$, $s_{i,t}$, and $v_{i,t}^{\mathrm{amp}}$ serve as
estimates of the bias $b_i$, phase $\phi_t$, and
$\tfrac12\Delta v_{\mathrm{phase},i}$, respectively.

\subsubsection{Amplitude Update}

Guided by Proposition~\ref{prop:mode_match_cancel},
we set $v_{i,t}^{\mathrm{amp}}$ to $\tfrac12\Delta v_{\mathrm{phase},i}$,
that is, half of the voltage change induced by a phase transition.
We therefore adapt this value online
using the largest per-step voltage change in the past window. 
Let $N_a$ denote the window length used for amplitude estimation.
If $N_a$ is sufficiently long to include a phase transition,
the most recent $N_a$ samples contain at least one recent transition,
and the largest one-step voltage change within this window
\(
\max_{0\le \ell <N_a-1}|v_{i,t-\ell}-v_{i,t-\ell-1}|
\)
provides an estimate of the phase-induced voltage shift $\Delta v_{\mathrm{phase},i}$.
The amplitude is therefore updated as
\begin{equation}
v_{i,t}^{\mathrm{amp}}
=
\tfrac12
\max_{0\le \ell <N_a-1}
\bigl|
v_{i,t-\ell}-v_{i,t-\ell-1}
\bigr|.
\label{eq:amp_update}
\end{equation}

\subsubsection{Bias Update}

Proposition~\ref{prop:minimax_deviation} shows that
the optimal bias can be determined
from the steady-state voltage extrema using~\eqref{eq:b_star_from_extrema}.
Since these extrema are not directly observable,
we approximate them from recent voltage measurements
over a horizon of $N_b$ samples.
If $N_b$ is sufficiently long relative to the typical interval between phase transitions,
the most recent $N_b$ samples capture the voltage range produced by the two phases.
The midpoint of this range therefore
provides an estimate of the bias correction term $\Delta \bm{b}$ in Proposition~\ref{prop:minimax_deviation}
\[
\Delta b_{i,t}
=
\tfrac12
\Bigl(
\max_{0\le\ell<N_b} v_{i,t-\ell}
+
\min_{0\le\ell<N_b} v_{i,t-\ell}
\Bigr)
-1.
\]
The bias is then updated incrementally as
\begin{equation}
v_{i,t+1}^{\mathrm{bias}}
=
v_{i,t}^{\mathrm{bias}}
-
\eta_b\, \Delta b_{i,t},
\label{eq:bias_update}
\end{equation}
where $\eta_b\in(0,1]$ is a smoothing gain.
Equation~\eqref{eq:b_star_from_extrema} corresponds to a full correction, i.e., $\eta_b=1$.
Here we instead use a smaller gain to reduce the influence of outdated extrema
that may remain in the finite window until replaced by newer samples.
A smaller $\eta_b$ helps prevent the bias update from overreacting to such stale extrema.

\subsubsection{Sign Selection}

The sign $s_{i,t}\in\{-1,+1\}$ acts as a local estimate of the unknown
data center phase $\phi_t$, enabling a decentralized approximation of
\eqref{eq:ideal_tracking} using only local voltage measurements.
We choose:
\begin{equation}
s_{i,t}
=
\begin{cases}
+1, & v_{i,t} > v_{i,t}^{\mathrm{bias}},\\
-1, & \text{otherwise}.
\end{cases}
\label{eq:sign_update}
\end{equation}

Because the bias approximates the midpoint between the two reference levels
and intra-phase fluctuations are typically much smaller than the level separation,
the voltage remains mostly above or below the bias within each phase,
providing a simple decentralized rule for reference selection.

\vspace{-0.3cm}
\subsection{The Decentralized Algorithm}
\label{sec:decentralized_algorithm}

The complete decentralized control algorithm, which combines the measurement-based reference adaptation above with the deadband and saturation operators of~\eqref{eq:dz_sat_def}, is summarized in Algorithm~\ref{alg:controller}: each bus $i$ maintains a window $\mathcal{W}_i$ of recent local voltage samples, from which it estimates the amplitude $v_i^{\mathrm{amp}}$, the sign $s_i$, and the bias $v_i^{\mathrm{bias}}$ via~\eqref{eq:amp_update}--\eqref{eq:sign_update}, and updates its reactive power injection $q_i$ accordingly.

Algorithm~\ref{alg:controller} does not depend on a particular workload, model scale, or cycle period.
Different AI training jobs or internal data center controls may lead to different phase durations and transition magnitudes. Since these changing power patterns are reflected in the most recent voltage measurements, the measurement-based implementation enables the controller to adapt effectively to variations in the power profile. As will be demonstrated in the case study, the proposed controller can accommodate disturbances from multiple data centers, as well as changes in the data center load profile induced by internal power-smoothing control.

We design hyperparameters that generalize across workloads with different cycle periods, without requiring separate manual retuning.
Specifically, we scale the estimation windows $N_a$ and $N_b$ proportionally to $T_{\mathrm{cycle}}$ so that they span a full cycle. In addition, we scale the bias-update rate inversely with $T_{\mathrm{cycle}}$ so that the bias refreshes a fixed number of times per cycle, while the smoothing gain $\eta_b$ is held fixed. The amplitude window $N_a$ then captures at least one phase transition, and the bias window $N_b$ captures the voltage range produced by both phases. A small $\eta_b$ reduces the influence of any single window measurement that may be stale. With this cycle-relative scaling, the same dimensionless parameter set applies across workloads with different $T_{\mathrm{cycle}}$. The implementation is also lightweight: each bus maintains an $O(N)$-sample rolling history with $N=\max(N_a,N_b)$ and uses only local arithmetic for the amplitude, sign, and bias updates. Moreover, no inter-bus communication is required, so the per-bus cost is independent of network size.

\begin{remark}[Compatibility with conventional and smoothed loads]
The controller needs no prior knowledge of whether a bus hosts a data center. When a bus has no abrupt transitions, as with conventional slowly varying load or a fully smoothed data center, $v_{i,t}^{\mathrm{amp}}$ falls to zero. The two reference levels $v_{i,t}^{\mathrm{bias}}\pm v_{i,t}^{\mathrm{amp}}$ then collapse into one, and the controller reduces to the standard fixed-reference droop baseline~\eqref{eq:droop_fixed}. When a data center smooths its load only partially, $v_{i,t}^{\mathrm{amp}}$ shrinks instead, and the reference switches by a correspondingly smaller amount. The same controller can therefore be applied, without retuning, to buses with no data center as well as to buses hosting data centers that perform load smoothing.
\end{remark}

\vspace{-0.3cm}
\section{Case Studies}

We evaluate the proposed controller on the IEEE 33-bus distribution feeder~\cite{baran1989network}. We conduct case studies that cover a set of real AI training power traces. We also demonstrate the performance of the proposed method under scenarios with multiple data centers and internal load smoothing.

The per-unit bases for power and voltage are $10~\mathrm{MVA}$ and $12.66~\mathrm{kV}$, respectively. The data center is placed at bus~22. The nominal bus loads follow the standard IEEE 33-bus data. The normalized GPU power trace is scaled so that its swing induces a $\pm4.5\%$ voltage deviation, just below the $\pm5\%$ limit. Voltage and active power injections update at $\Delta t_{\mathrm{sim}} = 0.1~\mathrm{s}$; reactive power injections are held constant within a slower control update period $\Delta t_{\mathrm{ctrl}} = 1.0~\mathrm{s}$, consistent with practical inverter-based control~\cite{yuan2017multi}. Bus voltages are simulated using the LinDistFlow model~\cite{baran1989network, zhu2015fast}. We compare the proposed switching-reference method against the fixed-reference droop controller in~\eqref{eq:droop_fixed}, which uses a constant reference of $1~\mathrm{p.u.}$ Both methods share the same linear droop gain $\mathbf{K}$, a symmetric deadband $\delta_i = 0.02~\mathrm{p.u.}$, and an action saturation $\Delta_i = 0.01~\mathrm{p.u.}$, consistent with IEEE~1547~\cite{ieee1547}.

We report two metrics throughout this section, both computed after a warmup of 6 workload cycles: the \emph{voltage violation} is the root-mean-square (RMS) and the peak of the band excess $\max\{|v_{i,t}-1|-0.05,\,0\}$ over all buses and time steps, and the \emph{control effort} is the analogous root-mean-square of the reactive power increment $\Delta q_{i,t}$. The code reproducing these experiments is publicly available at \url{https://github.com/yan-mingyuan/switchref}.

\subsection{Performance Across AI Training Workloads}
\label{sec:perf_across_workloads}

\begin{figure}[t]
\centering
\includegraphics[width=0.95\linewidth]{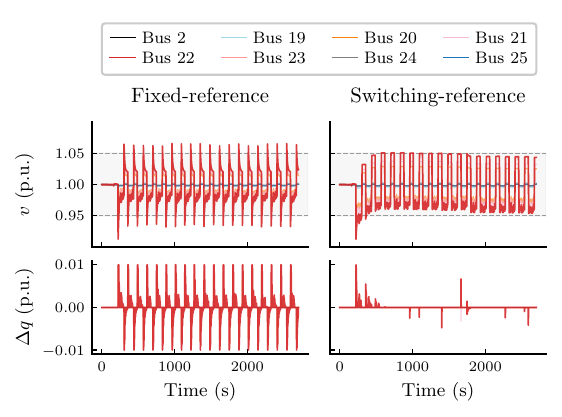}
\vspace{-0.5cm}
\caption{
Single data center scenario: voltage trajectory $\bm v$ (top) and control action $\Delta\bm q$ (bottom) under fixed-reference (left) and switching-reference (right) control. The proposed control substantially reduces control effort while maintaining voltage within the admissible $\pm5\%$ voltage-deviation band.
\vspace{-0.5cm}
}
\label{fig:exp_single}
\end{figure}

\begin{figure*}[t]
\centering
\includegraphics[width=0.55\textwidth]{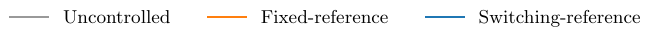}\\[-0.2cm]
\subfloat[4$\times$H200\vspace{-0.25cm}]{\includegraphics[width=0.95\textwidth]{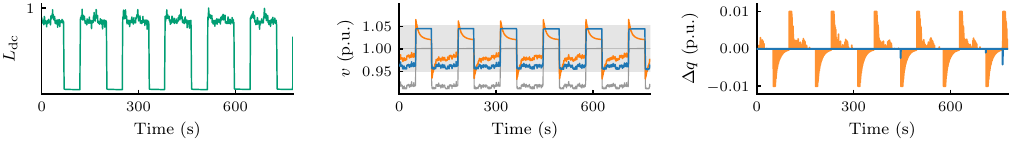}}\\
\vspace{-0.1cm}
\subfloat[4$\times$A100\vspace{-0.25cm}]{\includegraphics[width=0.95\textwidth]{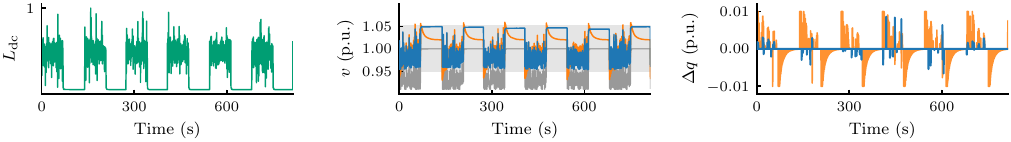}}\\
\vspace{-0.1cm}
\subfloat[4$\times$L40S\vspace{-0.25cm}]{\includegraphics[width=0.95\textwidth]{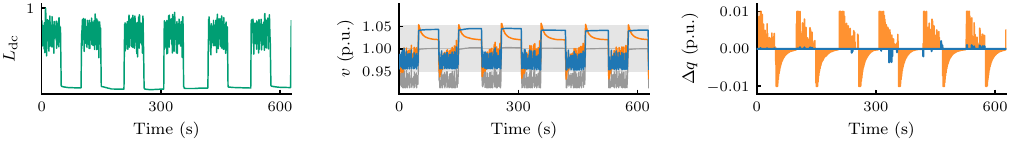}}\\
\vspace{-0.1cm}
\subfloat[At-scale H100\vspace{-0.2cm}]{\includegraphics[width=0.95\textwidth]{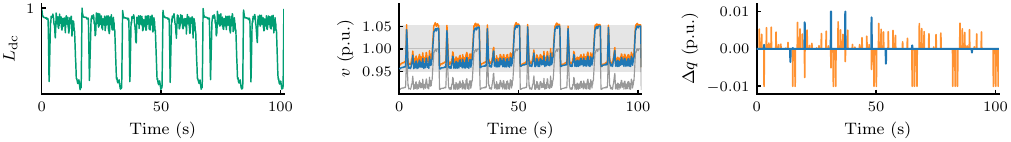}}
\caption{\small Generalization across AI training power traces: the normalized data center load $L_{\mathrm{dc}}$ (left); the data center bus voltage $v$ under no control, fixed-reference, and switching-reference control (middle); the control $\Delta q$ under fixed- and switching-reference control (right). The proposed switching-reference control reduces voltage violation and control effort consistently across diverse training workloads.\vspace{-0.5cm}}
\label{fig:generalization}
\end{figure*}

We evaluate the controller on real AI training power traces: 4-GPU node aggregates from $4\times$H200, $4\times$A100, and $4\times$L40S training jobs on a high-performance computing cluster, and the at-scale H100 trace of Fig.~\ref{fig:train_power} reported in~\cite{choukse2025power}. These traces share the synchronized compute/communication structure of large-model training, in which accelerators across the cluster execute the same iteration step in lockstep, and therefore share this two-phase structure at data center scale~\cite{choukse2025power}. While the specific model and training configuration vary across traces, the underlying two-phase pattern is preserved across all of them; only the cycle period differs.

Fig.~\ref{fig:exp_single} compares voltages and control actions across selected buses in the distribution system under both fixed-reference and switching-reference control, where the data center power traces are recorded during distributed training of the LLaMA-2-70B model~\cite{touvron2023llama} with a QLoRA configuration~\cite{dettmers2023qlora} on the 4$\times$H200 GPUs. Under fixed-reference control, actions repeatedly reach the saturation limit at each phase transition, and the voltage shows persistent oscillations synchronized with the load. The proposed switching-reference control substantially reduces both the magnitude and persistence of the actions after a few initial cycles, keeping the voltage close to nominal and rarely outside the admissible band.

\begin{table}[!t]
\caption{Performance across training traces for a single data center.}
\label{tab:metrics}
\centering
\setlength{\tabcolsep}{3pt}
\begin{tabular}{llccc}
\toprule
\multirow{2}{*}{Trace} & \multirow{2}{*}{Controller}
    & \multicolumn{2}{c}{Voltage violation (p.u.)} & Control effort (p.u.) \\
\cmidrule(lr){3-4}
& & RMS ($\times 10^{-4}$) & Peak ($\times 10^{-2}$) & RMS ($\times 10^{-4}$) \\
\midrule
\multirow{3}{*}{4$\times$H200}  & Uncontrolled & 50.53 & 3.91 & -- \\
                                & Fixed & 3.92 & 1.78 & 3.02 \\
                                & \textbf{Switching} & \textbf{0.23} & \textbf{0.10} & \textbf{0.15} \\
\midrule
\multirow{3}{*}{4$\times$A100}  & Uncontrolled & 34.07 & 3.99 & -- \\
                                & Fixed & 3.08 & 1.99 & 2.83 \\
                                & \textbf{Switching} & \textbf{0.13} & \textbf{0.04} & \textbf{0.45} \\
\midrule
\multirow{3}{*}{4$\times$L40S}  & Uncontrolled & 30.44 & 3.67 & -- \\
                                & Fixed & 2.40 & 1.80 & 2.98 \\
                                & \textbf{Switching} & $\bm{<\!0.01}$ & $\bm{<\!0.01}$ & \textbf{0.23} \\
\midrule
\multirow{3}{*}{\shortstack[l]{At-scale\\H100~\cite{choukse2025power}}} & Uncontrolled & 54.07 & 4.00 & -- \\
                                & Fixed & 3.42 & 1.09 & 3.38 \\
                                & \textbf{Switching} & \textbf{0.25} & \textbf{0.37} & \textbf{1.00} \\
\bottomrule
\end{tabular}
\vspace{-0.5cm}
\end{table}

Fig.~\ref{fig:generalization} shows the power traces across different training hardware and compares the corresponding voltages and control actions under fixed-reference and switching-reference control. Table~\ref{tab:metrics} quantifies this comparison in terms of voltage violation and control effort. Without reactive control, the load swing alone drives the RMS voltage violation to $30$--$54\times 10^{-4}$~p.u., well above either controller, confirming that voltage support is required. The switching reference suppresses voltage violations, sometimes eliminating them entirely, and reduces the control effort by an order of magnitude on average. The improvement is consistent across heterogeneous training hardware and cycle periods. For the at-scale AI training trace reported in~\cite{choukse2025power}, the switching reference again reduces both the RMS voltage violation and the control effort. The controller therefore generalizes across AI training workloads.

\subsection{Multiple Data Centers}
\label{sec:multi_dc}
A single distribution feeder may serve more than one large data center.
To evaluate this scenario, we consider two data centers on different feeder branches (buses~22 and~25), each running an independent training workload, and apply a decentralized switching-reference controller at each site with no coordination between the two.

Fig.~\ref{fig:exp_two} shows that the controller continues to suppress voltage deviations under the two simultaneously varying loads. As in the case with single data centers, the switching reference substantially reduces both the voltage violation, from $6.47$ to $0.88\times 10^{-4}$~p.u., and the control effort, from $4.81$ to $0.25\times 10^{-4}$~p.u. Each controller only leverages the local voltage measurements at its own bus. The decentralized design therefore remains effective with multiple data centers, requiring no communication between sites or central coordination.

\begin{figure}[t]
\centering
\includegraphics[width=\linewidth]{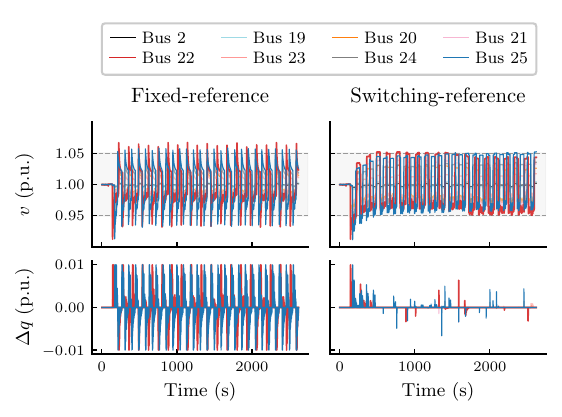}
\vspace{-0.5cm}
\caption{
Two data center scenario: voltage trajectory $\bm v$ (top) and control action $\Delta\bm q$ (bottom) under fixed-reference (left) and switching-reference (right) control. The proposed control remains effective with multiple data centers.
\vspace{-0.8cm}
}
\label{fig:exp_two}
\end{figure}

\vspace{-0.3cm}
\subsection{Internal Load Smoothing}
\label{sec:load_smoothing}
Data centers could conduct internal power smoothing using energy storage devices~\cite{choukse2025power, shi2017using}. Because such smoothing reshapes the load profile and may switch on during operation, a practical grid-side controller should remain effective across the resulting regime change without retuning. We therefore consider a case study in which the data center has energy storage sized to supply the maximum data center power for $60$~s (a typical backup-generator startup time). The net active power injection becomes $\hat p_{j,t}=p_{j,t}-z_{j,t}$, where $z_{j,t}$ follows a droop rule around a 100-s rolling average $\bar p_{j,t}$: $z_{j,t}=k_p(p_{j,t}-\bar p_{j,t})$ with droop gain $k_p=0.6$. The state of charge is constrained to $[0.93,0.97]$ around $95\%$ to preserve backup capability; the dispatch of energy storage is forced to zero at the bounds.

Fig.~\ref{fig:exp3_compare} shows the data center load profile, control actions, and voltage trajectory before and after load smoothing activates at $t\approx 1400$~s. After activation of load smoothing, the magnitude of the load fluctuations and the corresponding voltage deviations are smaller. The controller adapts to the modified load profile, producing only sparse, small actions. In particular, through the bias and amplitude update rules of Section~\ref{sec:decentralized_algorithm}, the controller adaptively recenters the voltage around $1$~p.u. This recentering reduces the maximum voltage deviation and improves the worst-case voltage regulation. The controller adapts across the regime change without retuning, demonstrating compatibility with data center load smoothing.

\begin{figure}[t]
\centering
\includegraphics[width=0.9\linewidth]{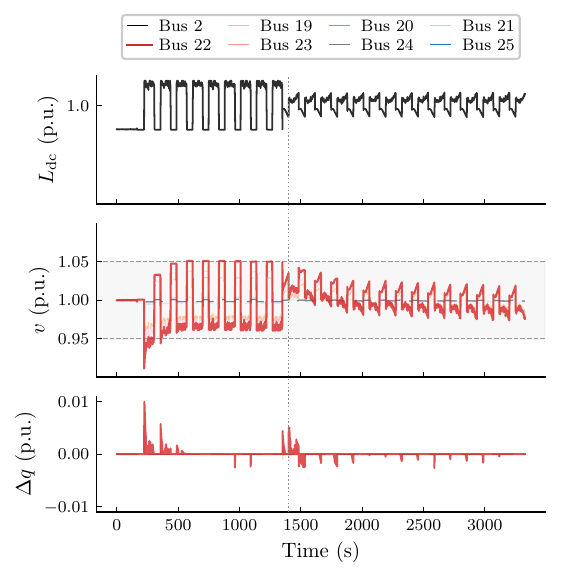}
\vspace{-0.2cm}
\caption{Compatibility with internal data center load smoothing.
Top to bottom: data center active power $\bm p$, voltage trajectory $\bm v$, and control action $\Delta\bm q$ under the proposed switching-reference controller.
Internal storage-based smoothing activates at $t\approx1400$~s. Switching-reference control adapts across the regime change without retuning, keeping the control actions sparse and small.
\vspace{-0.5cm}
}
\label{fig:exp3_compare}
\end{figure}

\vspace{-0.3cm}
\section{Conclusion}

This paper studies voltage control in distribution systems with AI-training data centers.
A decentralized switching-reference voltage control framework
is developed to exploit the two-phase structure of AI training workloads.
Conditions for voltage convergence are established,
and a measurement-based switching-reference design is developed
to reduce worst-case voltage deviations with low control effort.
The analysis is further extended to the practical controller with deadband and action saturation.
Case studies show that the proposed controller suppresses voltage violations, sometimes eliminating them entirely, and reduces the control effort by an order of magnitude compared with conventional droop control. The controller remains effective under multiple data centers and internal load smoothing, without requiring hyperparameter retuning or extra coordination across the distribution system.

\appendix

\subsection{Error Dynamics}
\label{app:error_dynamics}
Starting from the definition
$\bm{e}_t := \bm{v}_t - \bm{v}_t^{\mathrm{ref}},$
we derive the deviation recursion~\eqref{eq:linear_error_system}.
Recall the voltage model~\eqref{eq:lin_voltage} and the droop update~\eqref{eq:droop} (the latter rewritten using $\bm{e}_t$):
\[
\bm{v}_t
= \mathbf{R}\bm{p}_t + \mathbf{X}\bm{q}_t + \mathbf{1},
\qquad
\bm{q}_{t+1}
= \bm{q}_t - \mathbf{K}\bm{e}_t .
\]
Then
\begin{align}
\bm{e}_{t+1}
&= \bm{v}_{t+1}-\bm{v}_{t+1}^{\mathrm{ref}} \notag\\
&= \mathbf{R}\bm{p}_{t+1}+\mathbf{X}\bm{q}_{t+1}+\mathbf{1}-\bm{v}_{t+1}^{\mathrm{ref}} \notag\\
&= \mathbf{R}\bm{p}_{t+1}+\mathbf{X}(\bm{q}_t-\mathbf{K}\bm{e}_t)+\mathbf{1}-\bm{v}_{t+1}^{\mathrm{ref}} \notag\\
&= \mathbf{R}\bm{p}_{t+1}+\bigl(\bm{v}_t-\mathbf{R}\bm{p}_t-\mathbf{1}\bigr)-\mathbf{X}\mathbf{K}\bm{e}_t+\mathbf{1}-\bm{v}_{t+1}^{\mathrm{ref}} \notag\\
&= (\mathbf{I}-\mathbf{X}\mathbf{K})\bm{e}_t
+ \mathbf{R}(\bm{p}_{t+1}-\bm{p}_t)
- \bigl(\bm{v}_{t+1}^{\mathrm{ref}}-\bm{v}_t^{\mathrm{ref}}\bigr)\notag\\
&= (\mathbf{I}-\mathbf{X}\mathbf{K})\bm{e}_t
+ \mathbf{R}\Delta\bm{p}_t
- \Delta\bm{v}_t^{\mathrm{ref}}.\notag
\end{align}

\subsection{Gain Contraction}
\label{app:gain_to_norm}

\begin{lemma}
\label{lem:gain_contraction}
Assume $\mathbf{X}=\mathbf{X}^\top\succ0$ and the diagonal matrix $\mathbf{K}$ satisfies
$0\prec \mathbf{K} \prec 2\mathbf{X}^{-1}$.
Let $\mathbf{A}:=\mathbf{I}-\mathbf{X}\mathbf{K}$.
Then there exist a nonnegative matrix $\mathbf C$
and a constant $0<\epsilon<1$ such that
\[
|\mathbf{A}^k| \le \mathbf C\,\epsilon^k,
\qquad \forall k\ge0 .
\]
\end{lemma}

\begin{proof}
From $0\prec \mathbf{K}\prec 2\mathbf{X}^{-1}$ and $\mathbf{X}\succ0$, it follows that
$0\prec \mathbf{X}^{1/2}\mathbf{K}\mathbf{X}^{1/2}\prec 2\mathbf{I}$~\cite{cui2022decentralized}.
Hence all eigenvalues of
$\mathbf{I}-\mathbf{X}^{1/2}\mathbf{K}\mathbf{X}^{1/2}$ lie in $(-1,1)$.

Moreover,
$
\mathbf{A}
=
\mathbf{X}^{1/2}
(\mathbf{I}-\mathbf{X}^{1/2}\mathbf{K}\mathbf{X}^{1/2})
\mathbf{X}^{-1/2},
$
so $\mathbf A$ is similar to the symmetric matrix
$\mathbf{I}-\mathbf{X}^{1/2}\mathbf{K}\mathbf{X}^{1/2}$.
Let $\epsilon:=\rho(\mathbf A)<1$.
Then there exists an invertible matrix $\mathbf P$ such that
$\mathbf A=\mathbf P\bm\Lambda\mathbf P^{-1}$ with
$|\lambda_i|\le\epsilon$.

Using the elementwise inequality $|\mathbf M\mathbf N|\le |\mathbf M|\,|\mathbf N|$ and $|\bm\Lambda|^k \le \epsilon^k \mathbf I$ gives
\[
|\mathbf A^k|
=
|\mathbf P\bm\Lambda^k\mathbf P^{-1}|
\le
|\mathbf P|\,|\bm\Lambda|^k\,|\mathbf P^{-1}|
\le
|\mathbf P|\,|\mathbf P^{-1}|\,\epsilon^k .
\]
The claim follows by taking
$\mathbf C:=|\mathbf P|\,|\mathbf P^{-1}|$.
\end{proof}

\subsection{Proof of Theorem~\ref{thm:deviation_bound}}
\label{app:deviation_bound}
Let
\(
\bm e_t := \bm v_t-\bm v_t^{\mathrm{ref}}
\)
and
\(
\mathbf A := \mathbf I-\mathbf X\mathbf K.
\)
Then~\eqref{eq:linear_error_system} gives
\(
\bm e_{t+1}=\mathbf A\bm e_t+\bm d_t.
\)
By Lemma~\ref{lem:gain_contraction},
there exist $\mathbf C \ge \mathbf 0$ and $0<\epsilon<1$ such that
\(
|\mathbf A^k|\le \mathbf C\,\epsilon^k
\)
for all $k\ge0$.
Unrolling from $t_0$ gives
$\bm e_t=\mathbf A^{t-t_0}\bm e_{t_0}+\sum_{\tau=t_0}^{t-1}\mathbf A^{t-1-\tau}\bm d_\tau$.
Taking absolute values and applying the bound yields
$|\bm e_t|\le \mathbf C\,\epsilon^{t-t_0}|\bm e_{t_0}|+\sum_{\tau=t_0}^{t-1}\mathbf C\,\epsilon^{t-1-\tau}|\bm d_\tau|$.

If $|\bm d_t|\le\bar{\bm d}$ for all $t$, then
$
|\bm e_t|
\le
\mathbf C\,\epsilon^{t-t_0}|\bm e_{t_0}|
+
\mathbf C\sum_{\ell=0}^{t-1-t_0}\epsilon^\ell \bar{\bm d}.
$
Letting $t\to\infty$, the transient term $\mathbf C\,\epsilon^{t-t_0}|\bm e_{t_0}|$ vanishes because $0<\epsilon<1$, while the geometric series satisfies $\sum_{\ell=0}^{\infty}\epsilon^\ell=(1-\epsilon)^{-1}$. Substituting $\bm e_t=\bm v_t-\bm v_t^{\mathrm{ref}}$ then gives the deviation bound~\eqref{eq:deviation_bound_general},
\[
\limsup_{t\to\infty} |\bm v_t-\bm v_t^{\mathrm{ref}}|\leq \frac{1}{1-\epsilon}\mathbf C\,\bar{\bm d}.
\]

\subsection{Proof of Theorem~\ref{thm:implemented_bound}}
\label{app:convergence}
Let $\bm e_t := \bm v_t - \bm v_t^{\mathrm{ref}}$. The controller of Algorithm~\ref{alg:controller} is given by~\eqref{eq:db_sat_controller}, which augments the linear update $\bm q_{t+1}-\bm q_t = -\mathbf K\bm e_t$ with the operators of~\eqref{eq:dz_sat_def}. Throughout, $(\mathbf C,\epsilon)$ are as in Theorem~\ref{thm:deviation_bound}, $0\prec\mathbf K\prec 2\mathbf X^{-1}$, $|\bm d_t|\le\bar{\bm d}$ for all $t$, the small-gain condition~\eqref{eq:small_gain} holds, and $\bar{\bm e}:=\limsup_{t\to\infty}|\bm v_t-\bm v_t^{\mathrm{ref}}|$. We bound $\bar{\bm e}$ by tracking the residual each nonsmooth operator leaves on transition dynamics.

The deadband residual $\bm\sigma_t := \bm e_t - \mathrm{dz}_{\bm\delta}(\bm e_t)$ satisfies $\sigma_{t,i}=\mathrm{sign}(e_{t,i})\min(|e_{t,i}|,\delta_i)$ and $\mathrm{dz}_{\bm\delta}(\bm e_t)_i=\mathrm{sign}(e_{t,i})\max(0,|e_{t,i}|-\delta_i)$, hence
\[
|\bm\sigma_t|\le\bm\delta, \qquad |\mathrm{dz}_{\bm\delta}(\bm e_t)|=\max(\bm 0,|\bm e_t|-\bm\delta).
\]
The saturation residual is
\[
\bm\eta_t := \mathbf K\,\mathrm{dz}_{\bm\delta}(\bm e_t) - \mathrm{sat}_{\bm\Delta}\!\bigl(\mathbf K\,\mathrm{dz}_{\bm\delta}(\bm e_t)\bigr).
\]
Since saturation removes only the magnitude beyond $\bm\Delta$, $\bm\eta_t=\mathrm{sign}(\mathbf K\,\mathrm{dz}_{\bm\delta}(\bm e_t))\max(\bm 0,|\mathbf K\,\mathrm{dz}_{\bm\delta}(\bm e_t)|-\bm\Delta)$; combining with these bounds and $\mathbf K\succ\bm 0$ diagonal, and collapsing the nested $\max$, gives
\[
|\bm\eta_t|=\max(\bm 0,\,\mathbf K(|\bm e_t|-\bm\delta)-\bm\Delta).
\]
Then $\bm q_{t+1}-\bm q_t = -\mathbf K\bm e_t + \mathbf K\bm\sigma_t + \bm\eta_t$, and the LinDistFlow derivation of Appendix~\ref{app:error_dynamics} gives
\begin{equation}
\bm e_{t+1}=(\mathbf I-\mathbf{XK})\bm e_t+\bm d_t+\mathbf{XK}\bm\sigma_t+\mathbf X\bm\eta_t,
\label{eq:db_sat_recursion}
\end{equation}
with $\bm d_t$ as in~\eqref{eq:linear_error_system}. Unrolling, applying $|(\mathbf I-\mathbf{XK})^k|\le\mathbf C\epsilon^k$ (Lemma~\ref{lem:gain_contraction}), $|\bm\sigma_\tau|\le\bm\delta$, the identity for $|\bm\eta_t|$, and $\mathbf X,\mathbf{XK}\ge\bm 0$ ($\mathbf X$ nonnegative, $\mathbf K\succ\bm 0$ diagonal),
\begin{align*}
|\bm e_t|\le{}& \mathbf C\,\epsilon^{t-t_0}|\bm e_{t_0}|+\sum_{\tau=t_0}^{t-1}\mathbf C\,\epsilon^{t-1-\tau}(|\bm d_\tau|+\mathbf X\mathbf K\bm\delta)\\
&{}+\sum_{\tau=t_0}^{t-1}\epsilon^{t-1-\tau}\mathbf C\mathbf X\max(\bm 0,\mathbf K(|\bm e_\tau|-\bm\delta)-\bm\Delta).
\end{align*}
Taking $\limsup_{t\to\infty}$ and using $|\bm d_\tau|\le\bar{\bm d}$ and $\sum_\tau\epsilon^{t-1-\tau}\to(1-\epsilon)^{-1}$, the monotonicity and continuity of $\max(\bm 0,\mathbf K(|\bm e_t|-\bm\delta)-\bm\Delta)$ in $|\bm e_t|$ bound its limit superior by the value at $\bar{\bm e}$, giving
\begin{equation}
\bar{\bm e}\le\frac{\mathbf C(\bar{\bm d}+\mathbf X\mathbf K\bm\delta)}{1-\epsilon}+\frac{\mathbf C\mathbf X}{1-\epsilon}\max(\bm 0,\mathbf K(\bar{\bm e}-\bm\delta)-\bm\Delta).
\label{eq:db_sat_steady}
\end{equation}
By the definition of $S$, $\max(\bm 0,\mathbf K(\bar{\bm e}-\bm\delta)-\bm\Delta)=\bm\Pi_S\mathbf K\bar{\bm e}-\bm\Pi_S(\mathbf K\bm\delta+\bm\Delta)$. Substituting into~\eqref{eq:db_sat_steady} and collecting the $\bm\Pi_S$ terms yields the implicit form
\[
\bar{\bm e}\le\frac{\mathbf C(\bar{\bm d}+\mathbf X\mathbf K\bm\delta)}{1-\epsilon}+\frac{\mathbf C\mathbf X\bm\Pi_S(\mathbf K\bar{\bm e}-\mathbf K\bm\delta-\bm\Delta)}{1-\epsilon}.
\]
Rearranging,
\[
\Bigl(\mathbf I-\tfrac{\mathbf C\mathbf X\bm\Pi_S\mathbf K}{1-\epsilon}\Bigr)\bar{\bm e}\le\tfrac{\mathbf C[\bar{\bm d}+\mathbf X\mathbf K\bm\delta-\mathbf X\bm\Pi_S(\mathbf K\bm\delta+\bm\Delta)]}{1-\epsilon}.
\]
The matrix $\mathbf C\mathbf X\bm\Pi_S\mathbf K/(1-\epsilon)$ is nonnegative, and $\bm\Pi_S\mathbf K\le\mathbf K$ together with monotonicity of the spectral radius for nonnegative matrices and the small-gain condition~\eqref{eq:small_gain} gives $\rho(\mathbf C\mathbf X\bm\Pi_S\mathbf K)\le\rho(\mathbf C\mathbf X\mathbf K)<1-\epsilon$. A nonnegative matrix with spectral radius below $1$ yields a convergent Neumann series, so $\mathbf I-\mathbf C\mathbf X\bm\Pi_S\mathbf K/(1-\epsilon)$ is a nonsingular $M$-matrix with nonnegative inverse~\cite{berman1994nonnegative}; inverting it resolves the implicit bound into the closed form~\eqref{eq:implemented_bound} of Theorem~\ref{thm:implemented_bound}.

\subsection{Deadband-Only and Saturation-Only Bounds}

The bound of Theorem~\ref{thm:implemented_bound} simplifies when only one of the two nonsmooth operators is present. The small-gain condition~\eqref{eq:small_gain} was used in its proof only to absorb the saturation residual's implicit feedback; without saturation, this term vanishes and the contraction condition alone is sufficient:

\begin{corollary}[Deadband alone]
\label{cor:deadband}
The deadband-equipped controller $\bm q_{t+1}=\bm q_t-\mathbf K\,\mathrm{dz}_{\bm\delta}(\bm e_t)$ (no saturation) satisfies
\begin{equation}
\bar{\bm e}\le\frac{\mathbf C(\bar{\bm d}+\mathbf X\mathbf K\bm\delta)}{1-\epsilon}
\label{eq:db_bound}
\end{equation}
under the contraction condition $0\prec\mathbf K\prec 2\mathbf X^{-1}$.
\end{corollary}

\begin{proof}
With saturation absent, the residual $\bm\eta_t$ vanishes and~\eqref{eq:db_sat_recursion} loses its $\mathbf X\bm\eta_t$ term. The steps through~\eqref{eq:db_sat_steady} then leave only its first term: the $\max$ residual is gone, so neither the $M$-matrix inversion nor the small-gain condition is invoked.
\end{proof}

Dually, when the deadband is absent the bound retains the small-gain dependence:

\begin{corollary}[Saturation alone]
\label{cor:saturation}
The saturation-equipped controller $\bm q_{t+1}=\bm q_t-\mathrm{sat}_{\bm\Delta}(\mathbf K\bm e_t)$ (no deadband) satisfies
\begin{equation}
\bar{\bm e}\le\biggl(\mathbf I-\frac{\mathbf C\mathbf X\bm\Pi_S\mathbf K}{1-\epsilon}\biggr)^{\!-1}\frac{\mathbf C(\bar{\bm d}-\mathbf X\bm\Pi_S\bm\Delta)}{1-\epsilon}
\label{eq:sat_bound}
\end{equation}
if $0\prec\mathbf K\prec 2\mathbf X^{-1}$ and the small-gain condition~\eqref{eq:small_gain} holds.
\end{corollary}

\begin{proof}
Setting $\bm\delta=\bm 0$ in the proof of Theorem~\ref{thm:implemented_bound} eliminates the deadband residual $\bm\sigma_t$. The remaining steps then yield the stated bound; the small-gain condition is retained because the saturation residual still drives the implicit bound on $\bar{\bm e}$.
\end{proof}

\bibliographystyle{IEEEtran}
\bibliography{IEEEabrv, Reference}

@STRING{IEEE_J_PWRD       = "{IEEE} Trans. Power Del."}

@STRING{IEEE_J_PWRS       = "{IEEE} Trans. Power Syst."}

@article{baran1989network,
  author={Baran, M. E. and Wu, F. F.},
  title={Network reconfiguration in distribution systems for loss reduction and load balancing},
  journal=IEEE_J_PWRD,
  volume={4},
  number={2},
  pages={1401--1407},
  year={1989},
  doi={10.1109/61.25627}
}

@book{berman1994nonnegative,
  author={Berman, Abraham and Plemmons, Robert J.},
  title={Nonnegative Matrices in the Mathematical Sciences},
  series={Classics in Applied Mathematics},
  year={1994},
  publisher={SIAM},
  address={Philadelphia, PA}
}

@article{valverde2013mpc,
  author={Valverde, Gustavo and Van Cutsem, Thierry},
  title={Model predictive control of voltages in active distribution networks},
  journal=IEEE_J_SG,
  volume={4},
  number={4},
  pages={2152--2161},
  year={2013}
}

@article{zhu2015fast,
  author={Zhu, Hao and Liu, Hao Jan},
  title={Fast local voltage control under limited reactive power: Optimality and stability analysis},
  journal=IEEE_J_PWRS,
  volume={31},
  number={5},
  pages={3794--3803},
  year={2015},
  publisher={IEEE}
}

@article{shi2017using,
  author={Shi, Yuanyuan and Xu, Bolun and Wang, Di and Zhang, Baosen},
  title={Using battery storage for peak shaving and frequency regulation: Joint optimization for superlinear gains},
  journal=IEEE_J_PWRS,
  volume={33},
  number={3},
  pages={2882--2894},
  year={2017},
  publisher={IEEE}
}

@article{yuan2017multi,
  author={Yuan, Xiaoming and Hu, Jiabing and Cheng, Shijie},
  title={Multi-time scale dynamics in power electronics-dominated power systems},
  journal={Frontiers of Mechanical Engineering},
  volume={12},
  number={3},
  pages={303--311},
  year={2017},
  publisher={Springer}
}

@article{bernstein2018network,
  author={Baker, Kyri and Bernstein, Andrey and Dall'Anese, Emiliano and Zhao, Changhong},
  title={Network-Cognizant Voltage Droop Control for Distribution Grids},
  journal=IEEE_J_PWRS,
  volume={33},
  number={2},
  pages={2098--2108},
  year={2018},
  doi={10.1109/TPWRS.2017.2735379}
}

@misc{ieee1547,
  key={IEEE 1547},
  title={{IEEE} Standard for Interconnection and Interoperability of Distributed Energy Resources with Associated Electric Power Systems Interfaces},
  howpublished={IEEE Std 1547-2018 (Revision of IEEE Std 1547-2003)},
  organization={IEEE Standards Association},
  year={2018},
  month=apr
}

@article{hou2020decentralized,
  author={Hou, Shoulu and Ni, Wei and Zhao, Shuai and Cheng, Bo and Chen, Shiping and Chen, Junliang},
  title={Decentralized real-time optimization of voltage reconfigurable cloud computing data center},
  journal={{IEEE} Trans. Green Commun. Netw.},
  volume={4},
  number={2},
  pages={577--592},
  year={2020},
  publisher={IEEE}
}

@article{kaplan2020scaling,
  author={Kaplan, Jared and McCandlish, Sam and Henighan, Tom and Brown, Tom B. and Chess, Benjamin and Child, Rewon and Gray, Scott and Radford, Alec and Wu, Jeffrey and Amodei, Dario},
  title={Scaling laws for neural language models},
  journal={arXiv preprint arXiv:2001.08361},
  year={2020}
}

@book{national1996american,
  author={{ANSI/NEMA}},
  title={American National Standard for Electric Power Systems and Equipment-Voltage Ratings (60 Hz)},
  year={2020},
  publisher={National Electrical Manufacturers Association}
}

@inproceedings{sakalkar2020data,
  author={Sakalkar, Varun and Kontorinis, Vasileios and Landhuis, David and Li, Shaohong and De Ronde, Darren and Blooming, Thomas and Ramesh, Anand and Kennedy, James and Malone, Christopher and Clidaras, Jimmy and others},
  title={Data center power oversubscription with a medium voltage power plane and priority-aware capping},
  booktitle={Proceedings of the Twenty-Fifth International Conference on Architectural Support for Programming Languages and Operating Systems},
  pages={497--511},
  year={2020}
}

@article{zhang2021drlvoltvar,
  author={Zhang, Y. and Wang, X. and Wang, J. and Zhang, Y.},
  title={Deep reinforcement learning based {Volt-VAR} optimization in smart distribution systems},
  journal=IEEE_J_SG,
  volume={12},
  number={1},
  pages={361--371},
  year={2021}
}

@article{cui2022decentralized,
  author={Cui, Wenqi and Li, Jiayi and Zhang, Baosen},
  title={Decentralized safe reinforcement learning for inverter-based voltage control},
  journal={Electric Power Systems Research},
  volume={211},
  pages={108609},
  year={2022},
  publisher={Elsevier}
}

@article{hoffmann2022training,
  author={Hoffmann, Jordan and Borgeaud, Sebastian and Mensch, Arthur and Buchatskaya, Elena and Cai, Trevor and Rutherford, Eliza and Casas, Diego de Las and Hendricks, Lisa Anne and Welbl, Johannes and Clark, Aidan and others},
  title={Training compute-optimal large language models},
  journal={arXiv preprint arXiv:2203.15556},
  volume={10},
  year={2022}
}

@article{dettmers2023qlora,
  author={Dettmers, Tim and Pagnoni, Artidoro and Holtzman, Ari and Zettlemoyer, Luke},
  title={{QLoRA}: Efficient finetuning of quantized {LLMs}},
  journal={Advances in Neural Information Processing Systems},
  volume={36},
  pages={10088--10115},
  year={2023}
}

@article{touvron2023llama,
  author={Touvron, Hugo and Martin, Louis and Stone, Kevin and Albert, Peter and Almahairi, Amjad and Babaei, Yasmine and Bashlykov, Nikolay and Batra, Soumya and Bhargava, Prajjwal and Bhosale, Shruti and others},
  title={{Llama} 2: Open foundation and fine-tuned chat models},
  journal={arXiv preprint arXiv:2307.09288},
  year={2023}
}

@article{chen2025voltage,
  author={Chen, Yize and Zhang, Baosen},
  title={Voltage regulation in distribution systems with data center loads},
  journal={arXiv preprint arXiv:2507.06416},
  year={2025}
}

@article{choukse2025power,
  author={Choukse, Esha and Warrier, Brijesh and Heath, Scot and Belmont, Luz and Zhao, April and Khan, Hassan Ali and Harry, Brian and Kappel, Matthew and Hewett, Russell J. and Datta, Kushal and others},
  title={Power stabilization for {AI} training datacenters},
  journal={arXiv preprint arXiv:2508.14318},
  year={2025}
}

@article{grattafiori2024llama,
  author={Grattafiori, Aaron and others},
  title={The {Llama} 3 herd of models},
  journal={arXiv preprint arXiv:2407.21783},
  year={2024}
}

@article{chen2025electricity,
  author={Chen, Xin and Wang, Xiaoyang and Colacelli, Ana and Lee, Matt and Xie, Le},
  title={Electricity demand and grid impacts of {AI} data centers: Challenges and prospects},
  journal={arXiv preprint arXiv:2509.07218},
  year={2025}
}

@inproceedings{patel2024characterizing,
  author={Patel, Pratyush and Choukse, Esha and Zhang, Chaojie and Goiri, {\'I}{\~n}igo and Warrier, Brijesh and Mahalingam, Nithish and Bianchini, Ricardo},
  title={Characterizing power management opportunities for large language models in the cloud},
  booktitle={Proc. 29th ACM Int. Conf. Architectural Support for Programming Languages and Operating Systems (ASPLOS)},
  pages={207--222},
  year={2024}
}

@article{cui2025leveraging,
  author={Cui, Wenqi and Xie, Yiheng and Low, Steven and Wierman, Adam and Zhang, Baosen},
  title={Leveraging predictions in power system voltage control: An adaptive approach},
  journal={arXiv preprint arXiv:2509.09937},
  year={2025}
}

@techreport{iea_energy_and_ai,
  author={{International Energy Agency}},
  title={Energy and {AI}},
  institution={{International Energy Agency}},
  address={Paris, France},
  year={2025},
  url={https://www.iea.org/reports/energy-and-ai}
}

@article{xie2025enhancing,
  author={Xie, Yiheng and Cui, Wenqi and Wierman, Adam},
  title={Enhancing data center low-voltage ride-through},
  journal={arXiv preprint arXiv:2510.03867},
  year={2025}
}

@article{liang2026gpu,
  author={Liang, Zhirui and Chung, Jae-Won and Chowdhury, Mosharaf and Chen, Jiasi and Dvorkin, Vladimir},
  title={{GPU-to-Grid}: Voltage regulation via {GPU} utilization control},
  journal={arXiv preprint arXiv:2602.05116},
  year={2026}
}

@article{mao2026feedback,
  author={Mao, Yanyong and Mathieu, Johanna L. and Dvorkin, Vladimir},
  title={Online feedback optimization of energy storage to smooth data center grid impacts},
  journal={arXiv preprint arXiv:2603.20564},
  year={2026}
}

@article{agalgaonkar2014distribution,
  author={Agalgaonkar, Yashodhan P. and Pal, Bikash C. and Jabr, Rabih A.},
  title={Distribution Voltage Control Considering the Impact of {PV} Generation on Tap Changers and Autonomous Regulators},
  journal={IEEE Transactions on Power Systems},
  volume={29},
  number={1},
  pages={182--192},
  year={2014},
  doi={10.1109/TPWRS.2013.2279721}
}

\end{document}